\documentclass[conference,letterpaper,10pt]{IEEEtran}

\usepackage[numbers]{natbib}

\usepackage{graphicx}
\RequirePackage[usenames,dvipsnames]{xcolor}

\usepackage[cmex10]{amsmath}
\usepackage{amsfonts}
\usepackage{amsthm}

\newtheorem{theorem}{Theorem}
\newtheorem{lemma}[theorem]{Lemma}

\usepackage[section]{algorithm}
\usepackage[noend]{algpseudocode}
\usepackage{listings}

\usepackage{multicol}
\usepackage{multirow}

\usepackage[%
    font={small,sf},
    labelfont=bf,
    format=hang,    
    format=plain,
    margin=0pt,
    width=0.8\textwidth,
]{caption}
\usepackage[list=true]{subcaption}

\usepackage{url}
\usepackage{hyperref}
\hypersetup{
	colorlinks,
	linktocpage=true,
	linkcolor={red!50!black},
	citecolor={blue!50!black},
	urlcolor={MidnightBlue}
}

\usepackage{xspace}
\newcommand{\Figure}{Fig.\xspace}

\usepackage{cancel}

\usepackage{soul}
\setstcolor{red}

\definecolor{gray}{rgb}{0.5,0.5,0.5}
\newcommand{\hide}[1]{}
\newcommand{\ignore}[1]{}

\newcommand{\overbar}[1]{\mkern 1.5mu\overline{\mkern-1.5mu#1\mkern-1.5mu}\mkern 1.5mu}

\DeclareMathOperator*{\argmin}{arg\,min}

\begin{document}

\title{Parallel Algorithms for Generating Random Networks with Given Degree Sequences}

\author{
\IEEEauthorblockN{Maksudul Alam\IEEEauthorrefmark{1}\IEEEauthorrefmark{2} and
Maleq Khan\IEEEauthorrefmark{1}}
\IEEEauthorblockA{
\IEEEauthorblockA{\IEEEauthorrefmark{1}Network Dynamics and Simulation Science Laboratory, Virginia Bioinformatics Institute\\}
\IEEEauthorblockA{\IEEEauthorrefmark{2}Department of Computer Science\\}
Virginia Tech, Blacksburg, Virginia 24061 USA \\
Email: \texttt{\{maksud, maleq\}@vbi.vt.edu}}
}

\date{1 March 2014}

\maketitle

\begin{abstract}
Random networks are widely used for modeling and analyzing complex processes. 
Many mathematical models have been proposed to capture diverse real-world networks. 
One of the most important aspects of these models is degree distribution. 
Chung--Lu (CL) model is a random network model, which can produce networks with any given arbitrary degree distribution. 
The complex systems we deal with nowadays are growing larger and more diverse than ever.
Generating random networks with any given degree distribution consisting of billions of nodes and edges or more has become a necessity, which requires efficient and parallel algorithms. 
We present an MPI-based distributed memory parallel algorithm for generating massive random networks using CL model, which takes $O(\frac{m+n}{P}+P)$ time with high probability and $O(n)$ space per processor, where $n$, $m$, and $P$ are the number of nodes, edges and processors, respectively. The time efficiency is achieved by using a novel load-balancing algorithm. Our algorithms scale very well to a large number of processors and can generate massive power--law networks with one billion nodes and $250$ billion edges in one minute using $1024$ processors.
\end{abstract}




\IEEEpeerreviewmaketitle

\section{Introduction}
The advancements of modern technologies are causing a rapid growth of  complex systems. 
These systems, such as the Internet \cite{Siganos2003}, biological networks \cite{Girvan2002}, social networks \cite{Yang2011,Yang2012}, and various infrastructure networks \cite{Latora2005,Chassin2005}  are sometimes modeled by random graphs for the purpose of studying their behavior.
The study of these complex systems have significantly increased the interest in various random graph models such as Erd\H{o}s--R\'enyi (ER) \cite{Erdos1960}, small-world \cite{Watts1998}, Barab\'{a}si--Albert (BA) \cite{Barabasi1999}, Chung-Lu (CL) \cite{Chung2002}, HOT \cite{Carlson1999}, exponential random graph (ERGM) \cite{Robins2007}, recursive matrix (R-MAT)\cite{Chakrabarti2004}, and stochastic Kronecker graph (SKG) \cite{Leskovec2007,Leskovec2010} models. Among those models, the SKG model has been included in Graph500 supercomputer benchmark \cite{Graph500} due to its simple parallel implementation. The CL model exhibits the similar properties of the SKG model and further has the ability to generate a wider range of degree distributions \cite{Pinar2012}. To the best of our knowledge, there is no parallel algorithm for the CL model.

Analyzing a very large complex system requires generating massive random networks efficiently. As the interactions in a larger network lead to complex collective behavior, a smaller network may not exhibit the same behavior, even if both networks are generated using the same model. In \cite{Leskovec2008}, by experimental analysis, it was shown that the structure of larger networks is fundamentally different from small networks and many patterns emerge only in massive datasets.
Demand for large random networks necessitates efficient algorithms to generate such networks. However, even efficient sequential algorithms for generating such graphs were nonexistent until recently. 
Sequential algorithms are sometimes acceptable in network analysis with tens of thousands of nodes, but they are not appropriate for generating large graphs \cite{Batagelj2005}.
Although, recently some efficient sequential algorithms have been developed \cite{Chakrabarti2004,Batagelj2005,Leskovec2007,Miller2011}, these algorithms can generate networks with only millions of nodes in a reasonable time. But, generating networks with billions of nodes can take an undesirably longer amount of time. Further, a large memory requirement may even prohibit generating such large networks using these sequential algorithms. Thus, distributed-memory parallel algorithms are desirable in dealing with these problems. Shared-memory parallel algorithms also suffer from the memory restriction as these algorithms use the memory of a single machine. Also, most shared-memory systems are limited to only a few parallel processors whereas distributed-memory parallel systems are available with hundreds or thousands of processors.

In this paper, we present a time-efficient  MPI--based distributed memory parallel algorithm for generating random networks from a given sequence of expected degrees using the CL model. To the best of our knowledge, it is the first parallel algorithm for the CL model. The most challenging part of this algorithm is load-balancing. Partitioning the nodes with a balanced computational load is a non trivial problem. In a sequential setting, many algorithms for the load-balancing problem were studied \cite{Manne1995,Olstad1995,Pinar2004}. 
Some of them are exact and some are approximate. 
These algorithms uses many different techniques such as heuristic, iterative refinement, dynamic programming, and parametric search. All of these algorithms require at least $\Omega(n+P\log{n})$ time, where $n$, $P$ are the number of nodes and processors respectively.  To the best of our knowledge, there is no parallel algorithm for this problem. In this paper, we present a novel and efficient parallel algorithm for computing the balanced partitions in $O(\frac{n}{P}+P)$ time. The parallel algorithm for load balancing can be of independent interest and probably could be used in many other problems. Using this load balancing algorithm, the parallel algorithm for the CL model takes an overall runtime of $O(\frac{n+m}{P}+P)$ {w.h.p.}. The algorithm requires $O(n)$ space per processor. Our algorithm scales very well to a large number of processors and can generate a power-law networks with a billion nodes and $250$ billion edges in memory in less than a minute using $1024$ processors.
The rest of the paper is organized as follows. 
In Section~\ref{Section:PCL:chung-lu-model} we describe the problem and the efficient sequential algorithm. 
In Section~\ref{Section:TimeEfficient}, we present the parallel algorithm along with analysis of partitioning and load balancing. 
Experimental results showing the performance of our parallel algorithms are presented in Section~\ref{Section:PCL:exp}. We conclude in Section~\ref{Section:PCL:conclusion}.

\section{Preliminaries and Notations} \label{Section:PCL:prelim}
In the rest of the paper we use the following notations. We denote a network by $G(V,E)$, where $V$ and $E$ are the sets of vertices (nodes) and edges, respectively, with $m = |E|$ edges and $n = |V|$ vertices labeled as $0, 1, 2, \dots, n-1$. We use the terms \emph{node} and \emph{vertex} interchangeably. 

We develop parallel algorithm for the message passing interface (MPI) based distributed memory system, where the processors do not have any shared memory and each processor has its own local memory. The processors can exchange data and communicate with each other by exchanging messages. The processors have a shared file system and they read-write data files from the same external memory. However, such reading and writing of the files are done independently.

We use $K$, $M$ and $B$ to denote thousands, millions and billions, respectively; e.g., $2B$ stands for two billion.

\section{Chung--Lu Model and Efficient Sequential Algorithm}
\label{Section:PCL:chung-lu-model}
Chung--Lu (CL) model \cite{Chung2002} generates random networks from a given sequence of expected degrees. We are given $n$ nodes and a set of non-negative weights $w=(w_0,\ldots w_{n-1})$ assuming $\max_{i}w_{i}^2 < S$, where $S=\sum_{k}w_{k}$ \cite{Chung2002}. For every pair of nodes $i$ and $j$, edge $(i,j)$ is added to the graph with probability $p_{i,j}=w_{i}w_{j}/S$.
If no self loop is allowed, i.e., $i \neq j$, the expected degree of node $i$ is given by $\sum_{j}w_{i}w_{j}/S=w_{i} - w^2_{i}/S$. For massive graphs, where $n$ is very large, the average degree converges to $w_i$, thus $w_i$ represents the expected degree of node $i$ \cite{Miller2011}.

The na\"{\i}ve algorithm of CL model for an undirected graph with $n$ nodes takes each of the $n(n-1)/2$ possible node pairs $\{i,j\}$ and creates the edge with probability $p_{i,j}$, therefore requiring $O(n^2)$ time. An $O(n+m)$ algorithm was proposed in \cite{Miller2011}  to generate networks assuming $w$ is sorted in non-increasing order, where $m$ is the number of edges. It is easy to see that $O(n+m)$ is the best possible runtime to generate $m$ edges. The algorithm is based on the edge skipping technique introduced in \cite{Batagelj2005} for Erd\H{o}s--R\'enyi model. Adaptation of that technique leads to the efficient sequential algorithm in \cite{Miller2011}. The pseudocode of the algorithm is given in Algorithm~\ref{Algorithm:PCL:scl}, consisting of two procedures \Call{Serial--CL}{} and \Call{Create--Edges}{}. Note that we restructured Algorithm~\ref{Algorithm:PCL:scl} by defining procedure \Call{Create--Edges}{} to use it without any changes later in our parallel algorithm. Below we provide an overview and a brief description of the algorithm (for complete explanation and correctness see \cite{Miller2011}). 
\begin{algorithm}[t]
\caption{Sequential Chung--Lu Algorithm}
\label{Algorithm:PCL:scl}
\begin{algorithmic}[1]
\Procedure {Serial--CL}{$w$}
	\State{$S \leftarrow \sum_k w_k$}\label{Line:PCL:SCL:S}
	\State{$E \leftarrow $ \Call{Create--Edges}{$w$, $S$, $V$}}\label{Line:PCL:SCL:CreateEdges}
\EndProcedure
\Procedure{Create--Edges}{$w$, $S$, $V$}
	\State{$E \leftarrow \emptyset$}
	\ForAll {$i$ $\in$ $V$}\label{Line:PCL:SCL:foralli}
		\State $j \leftarrow i+1$\label{Line:PCL:SCL:j}, {$p \leftarrow \min(w_iw_j/S, 1)$}
		\While{$j < n$ and $p > 0$}\label{Line:PCL:SCL:while}
		    \If{$p \ne 1$}
			    \State{choose a random $r \in (0,1)$}
			    \State{$\delta \leftarrow \left\lfloor\log (r)/\log (1-p) \right\rfloor$}\label{Line:PCL:SCL:delta}		
			\Else
				\State{$\delta \leftarrow 0$}
		    \EndIf
		    \State{$v \leftarrow j + \delta$}\label{Line:PCL:SCL:v} \Comment{ skip $\delta$ edges}
		    \If{$v < n$}
		        \State{$q \leftarrow \min(w_{i}w_{v}/S, 1)$ }\label{Line:PCL:SCL:q}
		        \State{choose a random $r \in (0, 1)$}
		        \If{$r < q/p$} 
			        \State{$E \leftarrow E \cup \{i, v\}$}\label{Line:PCL:SCL:E}
	            \EndIf
	            \State{$p \leftarrow q$,\quad$j \leftarrow v + 1$}\label{Line:PCL:SCL:jv1}
		    \EndIf
		\EndWhile	
	\EndFor
	\State{\Return {$E$}}
\EndProcedure
\end{algorithmic}
\end{algorithm}

The algorithm starts at  \Call{Serial--CL}{}, which computes the sum $S$ and calls procedure \Call{Create--Edges}{$w,S,V$}, where $V$ is the entire set of nodes. For each node $i \in V$, the algorithm selects some random nodes $v$ from $[i+1, n-1]$, and creates the edges $(i, v)$. A na\"{\i}ve way to select the nodes $v$ from $[i+1, n-1]$ is: for each $j \in [i+1, n-1]$, select $j$ independently with probability $p_{i,j}=w_{i}w_{j}/S$, leading to an algorithm with run time $O(n^2)$. Instead, the algorithm skips the nodes that are not selected by a random skip length $\delta$ as follows. 
For each $i \in V$ (Line~\ref{Line:PCL:SCL:foralli}),
the algorithm starts with $j=i+1$ and computes a random skip length $\delta \leftarrow \left\lfloor \frac{\log (r)}{\log (1-p)}\right\rfloor$, where $r$ is a real number in $(0,1)$ chosen uniformly at random and $p = p_{i,j}=w_{i}w_{j}/S$. Then node $v$ is selected by skipping the next $\delta$ nodes (Line~\ref{Line:PCL:SCL:v}), and edge $(i,v)$ is selected with probability $q/p$, where $q = p_{i,v} = w_{i}w_{v}/S$ (Line~\ref{Line:PCL:SCL:q}--\ref{Line:PCL:SCL:E}). Then from the next node $j+v$, this cycle of skipping and selecting edges is repeated (while loop in Line~\ref{Line:PCL:SCL:while}--\ref{Line:PCL:SCL:jv1}). 
As we always have $i<j$ and no edge $(i,j)$ can be selected more than once, this algorithm does not create any self-loop or parallel edges.
As the set of weights $w$ is sorted in non-increasing order, for any node $i$, the probability $p_{i,j} = w_iw_j/S$ decreases monotonically with the increase of $j$. It is shown in \cite{Miller2011} that for any $i, j$, edge $(i,j)$ is included in $E$ with probability exactly $w_iw_j/S$, as desired, and that the algorithm runs in $O(n+m)$ time.

\section{Parallel Algorithm for the CL Model}
\label{Section:TimeEfficient}
Next we present our distributed memory parallel algorithm for the CL model.  Although our algorithm generates undirected edges, for the ease of discussion we consider $u$ as the \emph{source node} and $v$ as the \emph{destination node} for any edge $(u,v)$ generated by the procedure \Call{Create--Edges}{}. 
Let $T_u$ be the task of generating the edges from source node $u$ (Lines~\ref{Line:PCL:SCL:foralli}--\ref{Line:PCL:SCL:jv1} in Algorithm~\ref{Algorithm:PCL:scl}). It is easy to see that for any $u \ne u'$ tasks $T_u$ and $T_{u'}$ are independent, i.e., tasks $T_u$ and $T_{u'}$ can be executed independently by two different processors. Now execution of procedure \Call{Create--Edges}{$w, S, V$} is equivalent to executing the set of tasks $\{T_u:u \in V\}$.  Efficient parallelization of Algorithm~\ref{Algorithm:PCL:scl} requires:

\begin{itemize}
\item Computing the sum $S = \sum_{k=0}^{n-1}{w_{k}}$ in parallel
\item Dividing the  task of executing \Call{Create--Edges}{} into independent subtasks
\item Accurately estimating the computational cost for each task
\item Balancing load among the processors 
\end{itemize}

To compute the sum $S$ efficiently, a parallel sum operation is performed on  $w$ using $P$ processors, which takes $O(\frac{n}{P}+\log{P})$ time. To divide the task of executing procedure \Call{Create--Edges}{} into independent subtasks, the set of nodes $V$ is divided into $P$ disjoint subsets $V_1,V_2,\ldots, V_P$; that is, $V_i \subset V$, such that for any $i \ne j$, $V_i \cap V_j = \emptyset$ and $\bigcup_i V_i=V$. Then $V_i$ is assigned to processor $P_{i}$, and $P_{i}$ execute the tasks $\{T_u : u \in V_i\}$; that is, $P_{i}$ executes \Call{Create--Edges}{$w, S, V_i$}. 

Estimating and balancing computational loads accurately are the most challenging tasks. To achieve good speedup of the parallel algorithm, both tasks must also be done in parallel, which is a non-trivial problem. A good load balancing is achieved by properly partitioning the set of nodes $V$ such that the computational loads are equally distributed among the processors. We use two classes of partitioning schemes named consecutive partitioning (CP) and round-robin partitioning (RRP). In CP scheme consecutive nodes are assigned to each partition, whereas in RRP scheme nodes are assigned to the partitions in a round-robin fashion. 
The use of various partitioning schemes is not only interesting for understanding the performance of the algorithm, but also useful in analyzing the generated networks. It is sometimes desirable to generate networks on the fly and analyze it without performing disk I/O. Different partitioning schemes can be useful for different network analysis algorithms. 
Many network analysis algorithms require partitioning the graph into an equal number of nodes (or edges) per processor. Some algorithms also require the consecutive nodes to be stored in the same processor.
Before discussing the partitioning schemes in detail, we describe some formulations that are applicable to all of these schemes.

Let $e_{u}$ be the expected number of edges produced and $c_{u}$ be the computational cost in task $T_{u}$ for a \textit{source node} $u$. For the sake of simplicity, we assign one unit of  time to process a node or an edge. With $S = \sum_{v=0}^{n-1}w_{v}$, we have:
\begin{align}
e_{u} &= \textstyle \sum_{v=u+1}^{n-1}p_{u,v}  =\sum_{v=u+1}^{n-1}\frac{w_{u}w_{v}}{S} =\frac{w_u}{S}\sum_{v=u+1}^{n-1}w_{v} \label{Equation:PCL:eu}\\
c_{u} &= \textstyle  e_{u} + 1 \label{Equation:PCL:cu}
\end{align}

For two nodes $u,v\in V$ such that $u<v$, we have $c_{u} \geq c_{v}$ (see Lemma~\ref{lemma:eugeqev} in Appendix~\ref{Section:Appendix}). The expected number of edges generated by the tasks $\{T_u : u \in V_i\}$ is given by $m_{i}=\sum_{u \in V_{i}}e_{u}$. Note that the expected number of edges in the generated graph, i.e., the expected total number of edges generated by all processors is $m=|E|=\sum_{i=0}^{P-1} m_{i}= \sum_{u=0}^{n-1}e_{u}$.
The computational cost for processor $P_i$ is given by:
$\textstyle c(V_i) = \sum_{u \in V_{i}}c_{u} = \sum_{u \in V_{i}} (e_{u}+1) = m_i + |V_i|.
\label{Equation:PCL:pcost}
$
Therefore, the total cost for all processors is given by:
\begin{equation}
\textstyle
\sum_{i=0}^{P-1} c(V_i) = \sum_{i=0}^{P-1} \left( m_i + |V_i| \right) = m + n
\label{Equation:PCL:totalcost}
\end{equation}

\subsection{Consecutive Partitioning (CP)}
\label{Section:PCL:consecutive-partitioning}
Let partition $V_i$ starts at node $n_i$ and ends at node $n_{i+1}-1$, where $n_0=0$ and $n_P=n$, i.e., $V_i=\{n_i,n_i+1, \ldots,n_{i+1}-1 \}$ for all $i$. We say $n_{i}$ is the \textit{lower boundary} of partition $V_{i}$.  
A na\"{\i}ve way for partitioning $V$ is where each partition consists of an equal number of nodes, i.e., $|V_i|=\left\lceil\frac{n}{P}\right\rceil$ for all $i$. To keep the discussion neat, we simply use $\frac{n}{P}$. 
Although the number of nodes in each partition is equal, the computational cost among the processors is very imbalanced. For two consecutive partitions $V_i$ and $V_{i+1}$, $c(V_i) > c(V_{i+1})$ for all $i$ and the difference is at least $\frac{n^2}{SP^2} \overbar{W}_{i}\overbar{W}_{i+1}$, where $\overbar{W}_{i} = \frac{1}{|V_i|}\sum_{u\in V_i}{w_u}$, the average weight (degree) of the nodes in $V_i$ (see  Lemma~\ref{lemma:unp} in Appendix~\ref{Section:Appendix}). Thus $c(V_i)$ gradually decreases with $i$ by a large amount leading to a very imbalanced distribution of the computational cost. 

To demonstrate that na\"{\i}ve CP scheme leads to imbalanced distribution of computational cost, we generated two networks, both with one billion nodes: $i$) Erd\H{o}s--R\'enyi network with an average degree of $500$, and $ii$) Power--Law network with an average degree of $49.72$. We used $512$ processors, which is good enough for this experiment.
\begin{figure*}[t]
\centering
{\includegraphics[width=\textwidth]{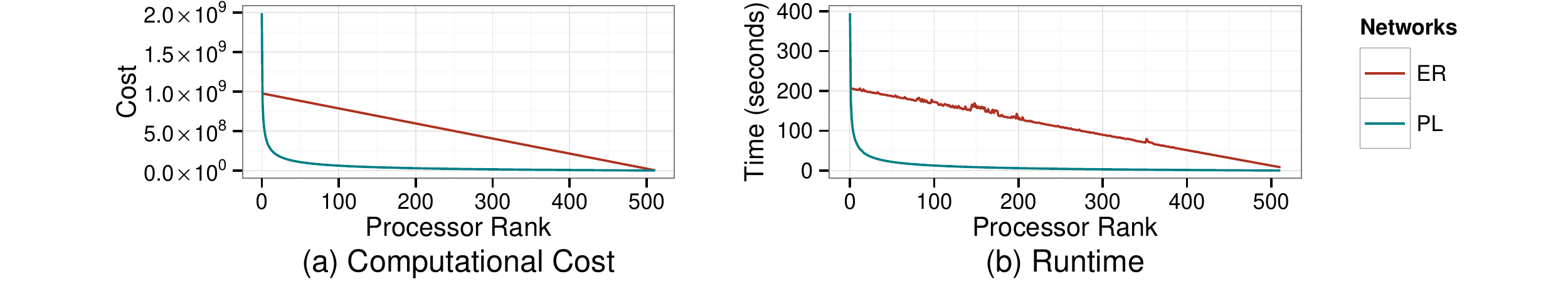}}
\caption{Computational cost and runtime in na\"{\i}ve CP scheme}
\label{Figure:PCL:UNP-LB}
\end{figure*}
\Figure~\ref{Figure:PCL:UNP-LB} shows the computational cost and runtime per processor. In both cases, the cost is not balanced. For power-law network the imbalance of computational cost is more prominent. Observe that the runtime is almost directly proportional to the cost, which justifies our choice of cost function. That is balancing the cost would also balance the runtime.

We need to find the partitions $V_{i}$ such that each partition has equal cost, i.e.,  $c(V_i)\approx\overbar{Z}$, where $\overbar{Z}=(m+n)/{P}$ is the average cost per processor. We refer such partitioning scheme as uniform cost partitioning (UCP).
Although determining the partition boundaries in the na\"{\i}ve scheme is very easy, finding the boundaries in UCP scheme is a non trivial problem and requires: (i) computing the cost $c_{u}$ for each node $u\in V$ and (ii) finding the boundaries of the partitions such that every partition has a cost of $\overbar{Z}$. 
Na\"{\i}vely computing costs for all nodes takes $O(n^2)$ as each node independently requires $O(n)$ time using Equation~\ref{Equation:PCL:cu} and \ref{Equation:PCL:eu}. A trivial parallelization achieves $O(n^2/P)$ time. 
However, our goal is to parallelize the computation of the costs in $O(n/{P}+\log{P})$ time.

Finding the partition boundaries such that the maximum cost of a partition is minimized is a well-known problem named \textit{chains-on-chains partitioning} (CCP) problem \cite{Pinar2004}. In CCP, a sequence of $P-1$ separators are determined to divide a chain of $n$ tasks with associated non-negative weights ($c_{u}$) into $P$ partitions so that the maximum cost in the partitions is minimized.
Sequential algorithms for CCP are studied quite extensively \cite{Manne1995,Olstad1995,Pinar2004}. Since these algorithms take at least $\Omega(n+P\log{n})$ time, using any of these sequential algorithms to find the partitions, along with the parallel algorithm for the CL model, does not scale well. To the best of our knowledge, there is no parallel algorithm for CCP problem. We present a novel parallel algorithm for determining the partition boundaries which takes $O(n/P+P)$ time in the worst case.
\begin{figure*}[t]
\centering
\includegraphics[width=\textwidth]{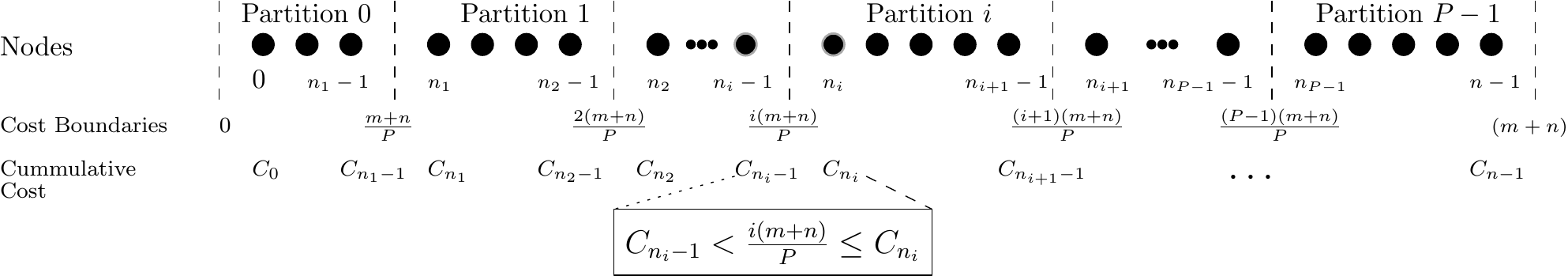}
\caption{Uniform cost partitioning (UCP) scheme}
\label{Figure:PCL:uep-partitions}
\end{figure*}

To determine the partition boundaries, instead of using $c_{u}$ directly, we use the cumulative cost $C_{u}=\sum_{v=0}^{u}c_{v}$. We call a partition $V_{i}$ a \textit{balanced partition} if the computational cost of $V_i$ is 
$c(V_i)=\sum_{u=n_{i}}^{n_{i+1}-1}c_{u}= C_{n_{i+1}-1} - C_{n_{i}-1}\approx\overbar{Z}$. 
Also note that for lower boundary $n_{i}$ of partition $V_{i}$ we have, $C_{n_i-1} < i\overbar{Z} \leq C_{n_i}$ for $0 < i \leq P-1$. Thus, we have:
\begin{equation}
\textstyle n_i = {\argmin_{u} \left( C_{u} \geq i\overbar{Z} \right)} \label{Equation:PCL:Partition-Boundary}
\end{equation}
In other words, a node $u$ with cumulative cost $C_{u}$ belongs to partition $V_i$ such that $i={\left\lfloor C_{u}/\overbar{Z}\right\rfloor}$. The partition scheme is shown visually in \Figure~\ref{Figure:PCL:uep-partitions}.

\begin{algorithm}[t]
\caption{Uniform Consecutive Partition}
\label{Algorithm:PCL:UCP}
\begin{algorithmic}[1]
    \Procedure{\textbf{UCP}}{$V$, $w$, $S$}
    \State{\Call{Calc--Cost}{$w$, $V$, $S$}}
    \State{\Call{Make--Partition}{$w$, $V$, $S$}}\label{Line:PCL:UCP:MakePartition}
    \EndProcedure
    \vspace{0.5em}
    \Procedure{\textbf{Calc--Cost}}{$w$, $V$, $S$}
    \State{$i \leftarrow$ processor id}
    \State{$s_{i} \leftarrow \sum_{u=i\frac{n}{P}}^{(i+1)\frac{n}{P}-1} w_u$}\label{Line:PCL:UCP:si}
    \State{\textbf{In Parallel:} $S_i \leftarrow \sum_{j=0}^{i-1}s_j $ }
    \label{Line:PCL:UCP:Si}
    \State{$u \leftarrow \frac{in}{P}$}\label{Line:PCL:UCP:u}
    \State{$\sigma_u \leftarrow S_{i}$}
    \State{$C_{u} \leftarrow e_{u} + 1 = \frac{w_u}{S} (S - \sigma_u - w_u)+1$}\label{Line:PCL:UCP:Cu}
    \For{$u=\frac{in}{P}+1$ to $\frac{(i+1)n}{P}-1$}\label{Line:PCL:UCP:foru1}
    \State{$\sigma_u \leftarrow \sigma_u+w_{u}$ }\label{Line:PCL:UCP:sigmau}
    \State{$e_{u} \leftarrow \frac{w_u}{S}(S - \sigma_u - w_{u})$}\label{Line:PCL:UCP:eu}
    \State{$C_{u} \leftarrow C_{u-1}+e_{u}+1$}\label{Line:PCL:UCP:Cu1}
    \EndFor
    \State{$z_{i} \leftarrow C_{\frac{(i+1)n}{P}-1}$}\label{Line:PCL:UCP:zi}
    
    \State{\textbf{In Parallel:} $ Z_i \leftarrow \sum_{j=0}^{i-1}z_j $ }\label{Line:PCL:UCP:Zi}
    \For{$u=\frac{in}{P}$ to $\frac{(i+1)n}{P}-1$}\label{Line:PCL:UCP:foru2}
    \State{$C_{u}=C_{u}+Z_{i}$}\label{Line:PCL:UCP:Cu2}
    \EndFor
    \EndProcedure
		\Procedure{\textbf{Make--Partition}}{$w$, $V$, $S$}
		\State{\textbf{In Parallel:} $Z \leftarrow \sum_{i=0}^{P-1} z_i$ }
		\label{Line:PCL:UCP:Z}
		\State{$\overbar{Z} \leftarrow Z/P$}\label{Line:PCL:UCP:Zbar}
		\State{\Call{Find--Boundaries}{$\frac{in}{P},\frac{(i+1)n}{P}-1,C,\overbar{Z}$}\label{Line:PCL:UCP:CallBoundary}}
		\ForAll{$n_k \in B_i$}\label{Line:PCL:UCP:fornk}
		\State{Send  $n_k$ to $P_{k}$ and $P_{k+1}$}
		\EndFor
		
		\State{Receive boundaries $n_{i}$ and $n_{i+1}$}\label{Line:PCL:UCP:MSG}
		
		\State{\Return{$V_{i}=[n_{i},n_{i+1}-1]$}}
\EndProcedure
\vspace{0.5em}
\Procedure{\textbf{Find--Boundaries}}{$s$, $e$, $C$, $\overbar{Z}$}
		\If{$\left\lfloor\frac{C_s}{\overbar{Z}}\right\rfloor=\left\lfloor\frac{C_e}{\overbar{Z}}\right\rfloor$}
		\Return \label{Line:28}
		\EndIf
		\State{$m \leftarrow \frac{(e+s)}{2}$\label{Line:29}}
		\If{$\left\lfloor\frac{C_{m}}{\overbar{Z}}\right\rfloor\neq\left\lfloor\frac{C_{m+1}}{\overbar{Z}}\right\rfloor$\label{Line:30}}		
		\State{$n_{\left\lfloor\frac{C_{m+1}}{\overbar{Z}}\right\rfloor} \leftarrow m+1$\label{Line:31}}
		\EndIf
		\State{\Call{Find--Boundaries}{$s, m, C, \overbar{Z}$}\label{Line:32}}
		\State{\Call{Find--Boundaries}{$m+1, e, C, \overbar{Z}$}\label{Line:33}}
\EndProcedure
\end{algorithmic}
\end{algorithm}

\textbf{Computing $C_u$ in Parallel}. Computing $C_u$ has two difficulties: i) for a node $u$, computing $c_{u}$ by using Equation~\ref{Equation:PCL:eu} and \ref{Equation:PCL:cu} directly is inefficient and ii) $C_{u}$ is dependent on $C_{u-1}$, which is hard to parallelize. To overcome the first difficulty, we use the following form of $e_u$ to calculate $c_u$. From Equation~\ref{Equation:PCL:eu} we have:
\begin{eqnarray}
e_{u} &=& \frac{w_u}{S}\sum_{v=u+1}^{n-1}w_{v}
		=\frac{w_u}{S} \left( \sum_{v=0}^{n-1}w_v-\sum_{v=0}^{u}w_v \right)\nonumber\\
		&=& \frac{w_u}{S} \left( \sum_{v=0}^{n-1}w_v-\sum_{v=0}^{u-1}w_v - w_u \right) \nonumber\\
c_u	&=& \textstyle e_{u}+1 \nonumber\\
&=&\frac{w_u}{S} \left( S - \sigma_{u} - w_u\right) +1 \left[\text{where }\sigma_{u}=\sum_{v=0}^{u-1}w_{v}\right]
\label{Equation:PCL:expected-edges}
\end{eqnarray}
Therefore, $c_{u}$ can be computed by successively updating $\sigma_{u}=\sigma_{u-1}+w_{u-1}$.

To deal with the second difficulty, we compute $C_{u}$ in several steps using procedure \Call{Calc--Cost}{} as shown in Algorithm~\ref{Algorithm:PCL:UCP} (see \Figure~\ref{Figure:PCL:uep-scheme} in Appendix~\ref{Section:Appendix} for a visual representation of the algorithm). In each processor, the partitioning algorithm starts with procedure \Call{UCP}{} that calculates the cumulative costs using procedure  \Call{Calc--Cost}{}. Then procedure \Call{Make--Partition}{} is used to compute the partitioning boundaries.
At the beginning of the \Call{Calc--Cost}{} procedure, the task of computing costs for the $n$ nodes are distributed among the $P$ processors equally, i.e., processor $P_{i}$ is responsible for computing costs for the nodes  from $i\frac{n}{P}$ to $(i+1)\frac{n}{P}-1$. Note that these are the nodes that processor $P_{i}$ works with while executing the partitioning algorithm to find the boundaries of the partitions. 

In Step \textbf{1} (Line~\ref{Line:PCL:UCP:si}), $P_{i}$ computes a partial sum $s_{i}=\sum_{u=\frac{in}{P}}^{\frac{(i+1)n}{P}-1}w_u$ independently of other processors. In Step~\textbf{2} (Line~\ref{Line:PCL:UCP:Si}),  \emph{exclusive prefix sum} $S_i = \sum_{j=0}^{i-1}s_j$ is calculated for all $s_{i}$ where $0 \leq i \leq P-1$ and $S_0=0$. This exclusive prefix sum can be computed in parallel in $O(\log P)$ time \cite{Sanders2006}. We have:
\begin{align*}
S_i &=\textstyle \sum_{j=0}^{i-1} s_j = \sum_{j=0}^{i-1} \sum_{u=\frac{jn}{P}}^{\frac{(j+1)n}{P}-1}w_u = \sum_{u=0}^{\frac{in}{P}-1}w_u = \sigma_{\frac{in}{P}}
\end{align*}
In Step~\textbf{3}, $P_i$ partially computes $C_{u}$, where $\frac{in}{P} \leq u < \frac{(i+1)n}{P}$. By assigning $\sigma_{\frac{in}{P}}=S_i$,  $C_{\frac{in}{P}}$ is determined partially using Equation~\ref{Equation:PCL:expected-edges} in constant time (Line~\ref{Line:PCL:UCP:Cu}). For each $u$, values of $\sigma_{u}$ , $e_{u}$ and $C_{u}$ are also determined in constant time (Line~\ref{Line:PCL:UCP:foru1}--\ref{Line:PCL:UCP:Cu1}), where $\frac{in}{P}+1 \leq u \leq \frac{(i+1)n}{P}-1$.  
After Step~\textbf{3}, we have $C_{u}=\sum_{v=\frac{in}{P}}^{u}c_v$. To get the final value of $C_{u}=\sum_{v=0}^{u}c_v$, the value  $\sum_{v=0}^{v=\frac{in}{P}-1}c_v$ needs to be added. For a processor $P_i$, let $z_{i}=C_{\frac{(i+1)n}{P}-1}=\sum_{v=\frac{in}{P}}^{\frac{(i+1)n}{P}-1}c_v$. In Step~\textbf{4} (Line~\ref{Line:PCL:UCP:Zi}), another exclusive parallel prefix sum operation is performed on $z_{i}$ so that $$\textstyle Z_i=\sum_{j=0}^{i-1}z_{j}=\sum_{j=0}^{i-1}\sum_{v=\frac{jn}{P}}^{\frac{(j+1)n}{P}-1}c_v=\sum_{v=0}^{\frac{in}{P}-1}c_v.$$ Note that $Z_{i}$ is exactly the value required to get the final cumulative cost $C_{u}$. In Step~\textbf{5} (Lines~\ref{Line:PCL:UCP:foru2}--\ref{Line:PCL:UCP:Cu2}), $Z_i$ is added to $C_{u}$ for $\frac{in}{P} \leq u \leq \frac{(i+1)n}{P}-1$. 

\vspace{0.5em}
\noindent\textbf{Finding Partition Boundaries in Parallel}. The partition boundaries are determined using Equation~\ref{Equation:PCL:Partition-Boundary}. The  procedure \Call{Make--Partition}{} generates the partition boundaries. In Line~\ref{Line:PCL:UCP:Z}, parallel sum is performed on $z_{i}$ to determine $Z=\sum_{0}^{P-1}z_{i}=\sum_{0}^{n-1}c_{u}=n+m$, the total cost and $\overbar{Z}=\frac{Z}{P}$ be the average cost per processor (Line~\ref{Line:PCL:UCP:Zbar}). 
\Call{Find--Boundaries}{} is called to determine the boundaries (Line~\ref{Line:PCL:UCP:CallBoundary}).
From  Equation~\ref{Equation:PCL:Partition-Boundary} it is easy to show that a partition boundary is found between two consecutive nodes $u$ and $u+1$, such that ${\left\lfloor C_{u}/\overbar{Z}\right\rfloor} \neq {\left\lfloor C_{u+1}/\overbar{Z}\right\rfloor}$. Node $u+1$ is the lower boundary of partition $V_i$, where $i={\left\lfloor C_{u+1}/\overbar{Z}\right\rfloor}$.
$P_{i}$ executes \Call{Find--Boundaries}{} from nodes ${in}/{P}$ to ${(i+1)n}/{P}-1$. 
\Call{Find--Boundaries}{} is a divide \& conquer based algorithm to find all the boundaries in that range efficiently using the cumulative costs . All the found boundaries are stored in a local list. In Line~\ref{Line:28}, it is determined whether the range contains any boundary. If the range does not have any boundary, i.e., if ${\left\lfloor C_{s}/\overbar{Z}\right\rfloor} = {\left\lfloor C_{e}/\overbar{Z}\right\rfloor}$, the algorithm returns immediately. Otherwise, it determines the middle of the range $m$ in Line~\ref{Line:29}. In Line~\ref{Line:30}, the existence of a boundary between $m$ and $m+1$ is evaluated. If $m+1$ is indeed a lower partition boundary, it is stored in local list in Line~\ref{Line:31}. In Line~\ref{Line:32} and \ref{Line:33}, \Call{Find--Boundaries}{} is called with the ranges $[s,m]$ and $[m+1,e]$ respectively.
Note that the range $\left[{in}/{P},{(i+1)n}/{P}-1\right]$ may contain none, one or more boundaries. Let $B_i$ be the set of those boundaries. 
Once the set of boundaries $B_i$, for all $i$, are determined, the processors exchange these boundaries with each other as follows. Node $n_{k}$, in some $B_i$, is the boundary between the partitions $V_k$ and $V_{k+1}$, i.e., $n_{k}-1$ is the upper boundary of $V_k$, and $n_{k}$ is the lower boundary of $V_{k+1}$. In  Line~\ref{Line:PCL:UCP:fornk}, for each $n_k$ in the range $\left[{in}/{P},{(i+1)n}/{P}-1\right]$, processor $P_{i}$ sends a boundary message containing $n_{k}$ to processors $P_{k}$ and $P_{k+1}$. Notice that each processor $i$ receives exactly two boundary messages from other processors (Line~\ref{Line:PCL:UCP:MSG}), and these two messages determine the lower and upper boundary of the $i$-th partition $V_i$. That is, now each processor $i$ has  partition $V_i$ and is ready to execute the parallel algorithm for the CL model with UCP scheme.

The runtime of parallel Algorithm~\ref{Algorithm:PCL:UCP}  is $O(\frac{n}{P}+P)$ as shown in Theorem \ref{thm:ucptime}.

\begin{theorem} \label{thm:ucptime}
The parallel algorithm for determining the partition boundaries of the UCP scheme runs in $O(\frac{n}{P}+P)$ time, where $n$ and $P$ are the number of nodes and processors, respectively.
\end{theorem}
\begin{proof}
The parallel algorithm for determining the partition boundaries is shown in  Algorithm~\ref{Algorithm:PCL:UCP}. 
For each processor, Line~\ref{Line:PCL:UCP:si} takes $O(\frac{n}{P})$ time. The exclusive parallel prefix sum operation requires $O(\log P)$ time in Line~\ref{Line:PCL:UCP:Si}. Lines~\ref{Line:PCL:UCP:u}--\ref{Line:PCL:UCP:Cu} take constant time. The for loop at Line~\ref{Line:PCL:UCP:foru1} iterates $\frac{n}{P}-1$ times. Each execution of the for loop takes constant time for Lines~\ref{Line:PCL:UCP:sigmau}--\ref{Line:PCL:UCP:Cu1}. Hence, the for loop at Line~\ref{Line:PCL:UCP:foru1} takes $O(\frac{n}{P})$ time. The prefix sum in Line~\ref{Line:PCL:UCP:Zi} takes $O(\log P)$ time. The for loop at Line~\ref{Line:PCL:UCP:foru2} takes $O(\frac{n}{P})$ time.

The parallel sum operation in Line~\ref{Line:PCL:UCP:Z} takes $O(\log P)$ time using \texttt{MPI\_Reduce} function.
For each processor $P_{i}$, $n_{k}$'s are determined in \Call{Find--Boundaries}{} on the range of $\left[{in}/{P},{(i+1)n}/{P}-1\right]$. Finding a single partition boundary on these $\frac{n}{P}$ nodes require $O(\log{\frac{n}{P}})$ time.
If the range contains $x$ partition boundaries, then it takes  $O(\min \left\{\frac{n}{P},x\log{\frac{n}{P}}\right\})$ time.
For each partition boundary $n_{k}$, processor $i$ sends exactly two messages to the processors $P_k$ and $P_{k-1}$. Thus each processor receives exactly two messages. There are at most $P$ boundaries in $[\frac{in}{P}$, $\frac{(i+1)n}{P}-1]$. Thus, in the worst case, a processor may need to send at most $2P$ messages, which takes $O(P)$ time. Therefore, the total time in the worst case is $O(\frac{n}{P}+\min \left\{\frac{n}{P},P\log{\frac{n}{P}}\right\}+ P)=O(\frac{n}{P}+P)$.
\qed
\end{proof}

\begin{figure}[b]
\centering
{\includegraphics[trim=0.6cm 0.6cm 1.1cm 0.65cm,width=\columnwidth]{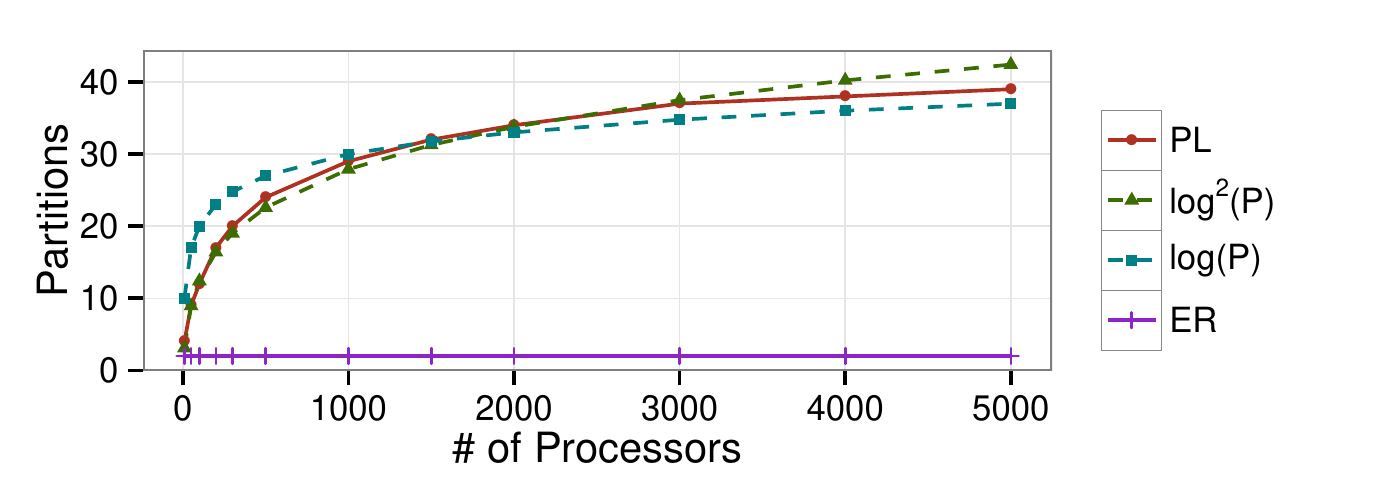}}
\caption{Maximum number of boundaries in a single processor}
\label{Figure:PCL:FirstBlock}
\end{figure}

Theorem \ref{thm:ucptime} shows the worst case runtime of $O(\frac{n}{P}+ P)$.
Notice that this bound on time is obtained considering the case that all $P$ partition boundaries $n_{k}$ can be in a single processor. However, in most real-world networks, it is an unlikely event, especially when the number of processors $P$ is large. Thus it is safe to say that for most practical cases, this algorithm will scale to a larger number of processors than the runtime analysis suggests. Now we experimentally show the number of partition boundaries found in the first partition for some popular networks.
For the ER networks, the maximum number of boundaries in a processor is  $2$, regardless of the number of processors. Even for the power--law networks, which has very skewed degree distribution, the maximum number of boundaries in a single processor is very small.
\Figure~\ref{Figure:PCL:FirstBlock} shows the maximum number of boundaries found in a single processor. Two fitted plots of $\log^2 P$ and $\log P$ is added in the figure for comparison. From the trend, it appears the maximum number of partition boundaries in a processor is somewhere between $O(\log P)$ and $O(\log^2{P})$. Since power--law has one of the most skewed degree distribution among real-world networks, we can expect the runtime to find partition boundaries to be approximately $O(\frac{n}{P}+\log^2P)$ time.

Using the UCP  scheme, our parallel algorithm for generating random networks with the CL model runs in $O(\frac{m+n}{P} + P)$  time as shown in Theorem~\ref{thm:chunglutime}. To prove Theorem~\ref{thm:chunglutime}, we need a bound on computation cost which is shown in Theorem~\ref{thm:comp_load}.

\begin{theorem}\label{thm:comp_load}
The computational cost in each processor is $O(\frac{m+n}{P})$  w.h.p.	
\end{theorem}
\begin{proof}
For each $u \in V_i$  and $v > u$, $(u,v)$  is a potential edge in processor $P_i$, and $P_i$ creates the edge with probability $p_{u,v}=\frac{w_uw_v}{S}$  where $S = \sum_{v\in V}{w_v}$. Let $x$ be the number of potential edges in $P_i$, and these potential edges are denoted by $f_1, f_2, \ldots, f_x$  (in any arbitrary order). Let $X_k$  be an indicator random variable such that $X_k = 1$  if $P_i$  creates $f_k$  and $X_k = 0$  otherwise. Then the number of edges created by $P_i$  is $X=\sum_{k=1}^{x}{X_k}$. 

As discussed in Section~\ref{Section:PCL:chung-lu-model}, generating the edges efficiently by applying the edge skipping technique is stochastically equivalent to generating each edge $(u,v)$ independently with probability $p_{u,v}=\frac{w_uw_v}{S}$. Let $\xi_{e}$ be the event that edge $e$ is generated. Regardless of the occurrence of any event $\xi_{e}$ with $e \ne (u,v)$, we always have $\Pr\{\xi_{(u,v)}\} = p_{u,v}=\frac{w_uw_v}{S}$. Thus, the events $\xi_{e}$ for all edges $e$ are mutually independent. Following the definitions and formalism given in Section~\ref{Section:PCL:consecutive-partitioning}, we have the expected number of edges created by $P_i$, denoted by $\mu$, as 
$$\textstyle\mu = E[X] = \sum_{u\in V_i} e_u = m_i.$$
Now we use the following standard Chernoff  bound for independent indicator random variables and for any $0< \delta < 1$, 
$$\textstyle\Pr\left\{X \ge (1+\delta) \mu\right\} \le e^{-\delta^2 \mu/3}.$$ 
Using this Chernoff  bound with $\delta = \frac{1}{2}$, we have 
$$\textstyle\Pr\left\{X \ge \frac{3}{2} m_i\right\} \le e^{- m_i/12} \le \frac{1}{m_i^3}$$ for any $m_i \ge 270$. We assume $m \gg P$  and consequently $m_i > P$  for all $i$. Now using the union bound, 
$$\textstyle\Pr\left\{X \ge \frac{3}{2} m_i\right\} \le m_i\frac{1}{m_i^3} = \frac{1}{m_i^2}$$
for all $i$  simultaneously. 
Then with probability at least $1 - \frac{1}{m_i^2}$, the computation cost $X + |V_i|$ is bounded by $\frac{3}{2}m_i+|V_i|=O(m_i+|V_i|)$. By
construction of the partitions by our algorithm, we have $O\left(m_i+|V_i|\right) = O\left(\frac{m +n}{P}\right)$. Thus the computation cost in all processors is $O\left( \frac{m +n}{P}\right)$ w.h.p.
\qed
\end{proof}

\begin{theorem}\label{thm:chunglutime}
Our parallel algorithm with UCP  scheme for generating random networks with the CL model runs in $O(\frac{m+n}{P} + P)$  time w.h.p.
\end{theorem}

\begin{proof}
Computing the sum $S$  in parallel takes $O\left(\frac{n}{P}+\log{P}\right)$  time. Using the UCP  scheme, node partitioning takes $O\left(\frac{n}{P} + P\right)$ time (Theorem~\ref{thm:ucptime}). In the UCP  scheme, each partition has $O\left(\frac{m+n}{P}\right)$  computation cost w.h.p. (Theorem~\ref{thm:comp_load}). Thus creating edges using procedure \Call{Create--Edges}{} requires $O\left(\frac{m+n}{P}\right)$  time, and the total time is $O\left(\frac{n}{P} + P + \frac{m+n}{P} \right)=O\left(\frac{m+n}{P} +P\right)$ w.h.p.
\qed
\end{proof}

\subsection{Round-Robin Partitioning (RRP)}
\label{Section:PCL:RRP}
In RRP scheme nodes are distributed in a round robin fashion.  Partition $V_i$ has the nodes $\langle i,i+P,i+2P,\ldots,i+kP \rangle $ such that $i + kP \leq n < i + (k+1)P$; i.e., $V_i=\{j|j \mod P =i\}$. In other words node $i$ is assigned to $V_{i \mod P}$. The number of nodes in each partition is almost equal, either $\lfloor \frac{n}{P} \rfloor$ or $\lceil \frac{n}{P} \rceil$. 

In order to compare the computational cost, consider two partitions $V_{i}$ and $V_{j}$ with $i < j$. Now, for the $x$-th nodes in these two partitions,  we have:  $c_{i+(x-1)P} \geq c_{j+(x-1)P}$  as $i+(x-1)P< j+(x-1)P$ (see Lemma~\ref{lemma:eugeqev}). Therefore, $c(V_{i})=\sum_{u\in V_{i}}c_{u} \geq c(V_{j})=\sum_{u\in V_{j}}c_{u}$ and by the definition of RRP scheme, $|V_{i}| \geq |V_{j}|.$ The difference in cost between any two partitions is at most $w_0$, the maximum weight (see Lemma~\ref{lemma:rrp} in Appendix~\ref{Section:Appendix}). Thus RRP scheme provides quite good load balancing. However, it is not as good as the UCP scheme. It is easy to see that in the RRP scheme, for any two partitions $V_{i}$ and $V_{j}$ such that $i<j$, we have $c(V_{i}) > c(V_{j})$. But, by design, the UCP scheme makes the partition such that cost are equally distributed among the processors. 
Furthermore, although the RRP scheme is simple to implement and provides quite good load balancing, it has another subtle problem. In this scheme, the nodes of a partition are not consecutive and are scattered in the entire range leading to some serious efficiency issues in accessing these nodes. One major issue is that the locality of reference is not maintained leading to a very high rate of cache miss during the execution of the algorithm. This contrast of performance between UCP and RRP is even more prominent when the goal is to generate massive networks as shown by experimental results in Section~\ref{Section:PCL:exp}.

\section{Experimental Results}
\label{Section:PCL:exp}
In this section, we experimentally show the accuracy and performance of our algorithm. The accuracy of our parallel algorithms is demonstrated by showing that the generated degree distributions closely match the input degree distribution. The strong scaling of our algorithm shows that it scales very well to a large number of processors. We also present experimental results showing the impact of the partitioning schemes on load balancing and performance of the algorithm.

\begin{table*}[t]
\caption{Networks used in the experiments}
\label{Table:Networks}
\centering
\begin{tabular}{|l|l|l|l|}
\hline
Network  & Type & Nodes & Edges \\ 
\hline
PL & Power Law Network &  1\textbf{B} & 249\textbf{B}\\
ER &  Erd\H{o}s--R\'enyi Network &  1\textbf{M} & 200\textbf{M}\\
Miami \cite{Barrett2009} &  Contact Network &  2.1\textbf{M} & 51.4\textbf{B}\\
Twitter \cite{Yang2011} &  Real--World Social Network &  41.65\textbf{M} & 1.37\textbf{B}\\
Friendster \cite{Yang2012} &  Real--World Social Network &  65.61\textbf{M} & 1.81\textbf{B}\\
\hline
\end{tabular}
\end{table*}

\begin{figure*}[t]
\centering
{\includegraphics[width=\textwidth]{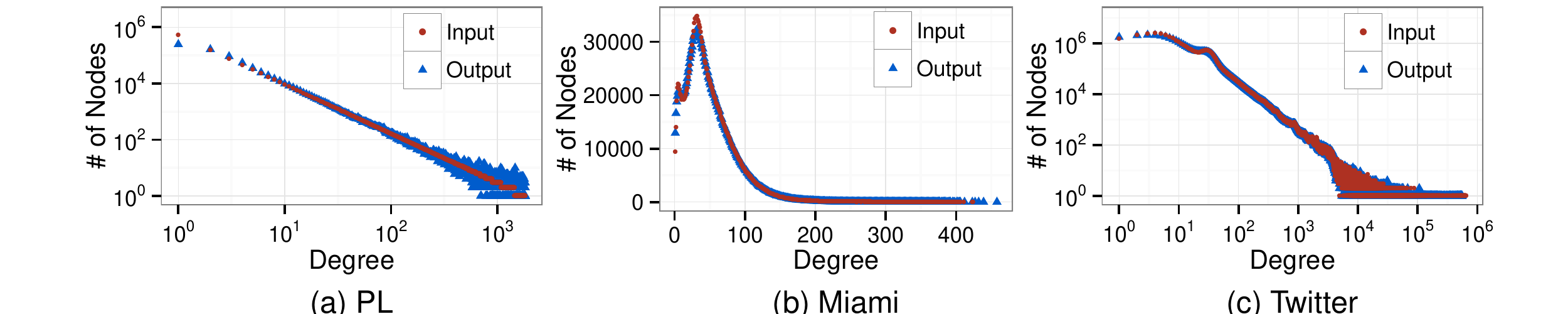}}
\caption{Degree distributions of input and generated degree sequences}
\label{Figure:PCL:DegreeDistributions}
\end{figure*}

\begin{figure*}[t]
\centering
{\includegraphics[width=\textwidth]{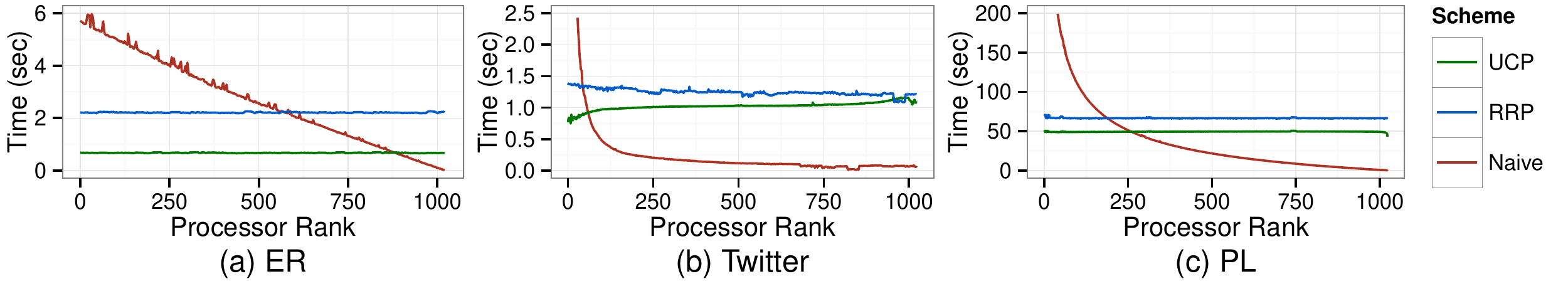}}
\caption{Comparison of partitioning schemes}
\label{Figure:PCL:PartitioningSchemes}
\end{figure*}

\begin{figure*}[t]
\centering
{\includegraphics[width=\textwidth]{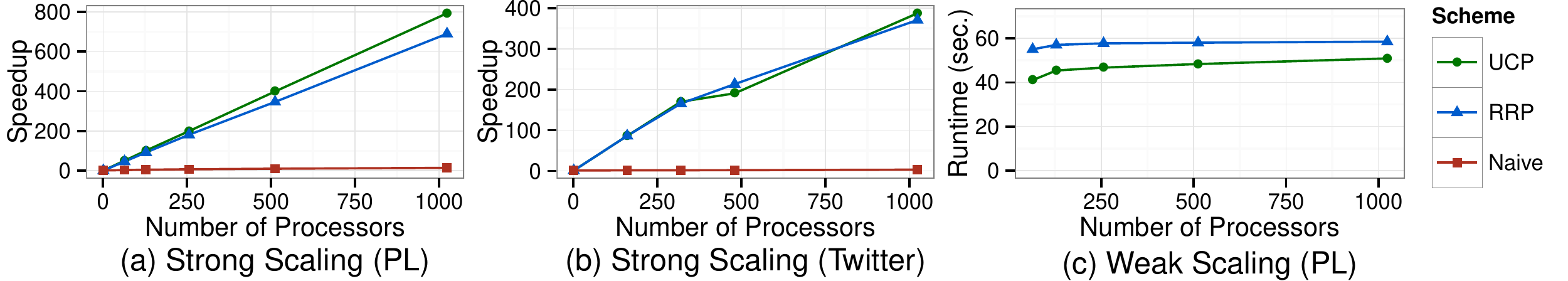}}
\caption{Strong and weak scaling of the parallel algorithms}
\label{Figure:PCL:cl-strongscaling}
\end{figure*}

\textbf{Experimental Setup.}
We used a $81$-node  HPC cluster for the experiments. Each node is powered by two octa-core  SandyBridge E5-2670 2.60GHz (3.3GHz Turbo) processors with $64$ GB memory. The algorithm is developed with MPICH2 (v1.7), optimized for QLogic InfiniBand cards. In the experiments, degree distributions of real-world and artificial random networks were considered. The list of networks is shown in Table~\ref{Table:Networks}. The runtime does not include the I/O time to write the graph into the disk.

\textbf{Degree Distribution of Generated Networks.}
\Figure~\ref{Figure:PCL:DegreeDistributions} shows the input and generated  degree distributions for PL, Miami, and Twitter networks (see Appendix~\ref{Section:OtherNetworks} for other networks).  As observed from the plots, the generated degree distributions closely follow the input degree distributions reassuring that our parallel algorithms generate random networks with given expected degree sequences accurately.

\textbf{Effect of Partitioning Schemes.} As discussed in Section~\ref{Section:PCL:consecutive-partitioning}, partitioning significantly affects load balancing and performance of the algorithm. We demonstrate the effects of  the partitioning schemes in terms of computing time in each processor as shown in  \Figure~\ref{Figure:PCL:PartitioningSchemes} using ER, Twitter, and PL networks.
Computational time fo na\"{\i}ve scheme  is skewed. For all the networks, the computational times for UCP and RRP stay almost constant in all processors, indicating good load-balancing.
RRP is little slower than UCP because the locality of references is not maintained in RRP, leading to high cache miss as discussed in Section~\ref{Section:PCL:RRP}.

\textbf{Strong and Weak Scaling.}
Strong scaling of a parallel algorithm shows it's performance with the increasing number of processors while keeping the problem size fixed. \Figure~\ref{Figure:PCL:cl-strongscaling} shows the speedup of na\"{\i}ve, UCP, and RRP partitioning schemes using PL  and Twitter networks. Speedups are measured as $\frac{T_s}{T_p}$, where $T_s$ and $T_p$ are the running time of the sequential and the parallel algorithm, respectively. The number of processors were varied from $1$ to $1024$. As \Figure~\ref{Figure:PCL:cl-strongscaling} shows, UCP and RRP achieve excellent linear speedups. Na\"{\i}ve scheme performs the worst as expected. The speedup of PL is greater than that of Twitter network. As Twitter is smaller than the PL network, the impact of the parallel communication overheads is higher contributing to decreased speedup. Still the algorithm to generate Twitter network has a speedup of $400$ using $1024$ processors. 

The weak scaling measures the performance of a parallel algorithm when the input size per processor remains constant. For this experiment, we varied the number of processors from $16$ to $1024$. For $P$ processors, a PL network with $10^6P$ nodes and $10^8P$ edges is generated. Note that weak scaling can only be performed on artificial networks. \Figure~\ref{Figure:PCL:cl-strongscaling}(c) shows the weak scaling for UCP and RRP schemes using PL networks. Both RRP and UCP show very good weak scaling with almost constant runtime.

\textbf{Generating Large Networks.}
The primary objective of the parallel algorithm is to generate massive random networks. Using the algorithm with UCP scheme,  we have generated power law networks with one billion nodes and $249$ billion edges in one minute using $1024$ processors with a speedup of about $800$.

\section{Conclusion}
\label{Section:PCL:conclusion}
We have developed an efficient parallel algorithm for generating massive networks with a given degree sequence using the Chung--Lu model. The main challenge in developing this algorithm is load balancing. To overcome this challenge, we have developed a novel parallel algorithm for balancing computational loads that results in a significant improvement in efficiency. We believe that the presented parallel algorithm for the Chung--Lu model will prove useful for modeling and analyzing emerging massive complex systems and uncovering patterns that emerges only in massive networks. As the algorithm can generate networks from any given degree sequence, its application will encompass a wide range of complex systems.

\bibliographystyle{IEEEtranN}
\bibliography{references}

\begin{thebibliography}{25}
\providecommand{\natexlab}[1]{#1}
\providecommand{\url}[1]{#1}
\csname url@samestyle\endcsname
\providecommand{\newblock}{\relax}
\providecommand{\bibinfo}[2]{#2}
\providecommand{\BIBentrySTDinterwordspacing}{\spaceskip=0pt\relax}
\providecommand{\BIBentryALTinterwordstretchfactor}{4}
\providecommand{\BIBentryALTinterwordspacing}{\spaceskip=\fontdimen2\font plus
\BIBentryALTinterwordstretchfactor\fontdimen3\font minus
  \fontdimen4\font\relax}
\providecommand{\BIBforeignlanguage}[2]{{%
\expandafter\ifx\csname l@#1\endcsname\relax
\typeout{** WARNING: IEEEtranN.bst: No hyphenation pattern has been}%
\typeout{** loaded for the language `#1'. Using the pattern for}%
\typeout{** the default language instead.}%
\else
\language=\csname l@#1\endcsname
\fi
#2}}
\providecommand{\BIBdecl}{\relax}
\BIBdecl

\bibitem[Siganos et~al.(2003)Siganos, Faloutsos, Faloutsos, and
  Faloutsos]{Siganos2003}
G.~Siganos, M.~Faloutsos, P.~Faloutsos, and C.~Faloutsos, ``Power laws and the
  as-level internet topology,'' \emph{IEEE/ACM Tran. on Networking}, 2003.

\bibitem[Girvan and Newman(2002)]{Girvan2002}
M.~Girvan and M.~Newman, ``{Community structure in social and biological
  networks},'' \emph{Proc. of the Nat. Aca. of Sci. of the USA}, 2002.

\bibitem[Yang and Leskovec(2011)]{Yang2011}
J.~Yang and J.~Leskovec, ``Patterns of temporal variation in online media,'' in
  \emph{Proc. of the 4th ACM Intel. Conf. on Web Search and Data Mining}, 2011.

\bibitem[Yang and Leskovec(2012)]{Yang2012}
------, ``Defining and evaluating network communities based on ground-truth,''
  in \emph{Proc. of the ACM SIGKDD Workshop}, 2012.

\bibitem[Latora and Marchiori(2005)]{Latora2005}
V.~Latora and M.~Marchiori, ``Vulnerability and protection of infrastructure
  networks,'' \emph{Phys. Rev. E}, 2005.

\bibitem[Chassin and Posse(2005)]{Chassin2005}
D.~Chassin and C.~Posse, ``{Evaluating North American electric grid reliability
  using the Barabasi-Albert network model},'' \emph{Physica A}, 2005.

\bibitem[Erd\"{o}s and R\'{e}nyi(1960)]{Erdos1960}
P.~Erd\"{o}s and A.~R\'{e}nyi, ``On the evolution of random graphs,'' in
  \emph{Publications of the Mathematical Institute of the Hungarian Academy of
  Sciences}, 1960.

\bibitem[Watts and Strogatz(1998)]{Watts1998}
D.~Watts and S.~Strogatz, ``{Collective dynamics of `small-world' networks},''
  \emph{Nature}, 1998.

\bibitem[Barab{\'{a}}si and Albert(1999)]{Barabasi1999}
A.-L. Barab{\'{a}}si and R.~Albert, ``{Emergence of scaling in random
  networks},'' \emph{Science}, 1999.

\bibitem[Chung and Lu(2002)]{Chung2002}
F.~Chung and L.~Lu, ``{Connected Components in Random Graphs with Given
  Expected Degree Sequences},'' \emph{Annals of Combinatorics}, 2002.

\bibitem[Carlson and Doyle(1999)]{Carlson1999}
J.~Carlson and J.~Doyle, ``Highly optimized tolerance: a mechanism for power
  laws in designed systems,'' \emph{Phys. Rev. E}, 1999.

\bibitem[Robins et~al.(2007)Robins, Pattison, Kalish, and Lusher]{Robins2007}
G.~Robins, P.~Pattison, Y.~Kalish, and D.~Lusher, ``An introduction to
  exponential random graph (p*) models for social networks social networks,''
  \emph{Soc. Net.}, 2007.

\bibitem[Chakrabarti et~al.(2004)Chakrabarti, Zhan, and
  Faloutsos]{Chakrabarti2004}
D.~Chakrabarti, Y.~Zhan, and C.~Faloutsos, ``{R-MAT: A Recursive Model for
  Graph Mining},'' in \emph{Fourth SIAM Intl. Conf. on Data Mining}, 2004.

\bibitem[Leskovec and Faloutsos(2007)]{Leskovec2007}
J.~Leskovec and C.~Faloutsos, ``Scalable modeling of real graphs using
  kronecker multiplication,'' in \emph{Proc. of the 24th Intl. Conf. on Machine
  Learning}, 2007.

\bibitem[Leskovec et~al.(2010)Leskovec, Chakrabarti, Kleinberg, Faloutsos, and
  Ghahramani]{Leskovec2010}
J.~Leskovec, D.~Chakrabarti, J.~Kleinberg, C.~Faloutsos, and Z.~Ghahramani,
  ``Kronecker graphs: An approach to modeling networks,'' \emph{J. Mach. Learn.
  Res.}, 2010.

\bibitem[Graph500(2015)]{Graph500}
Graph500. (2015) {Graph500}. \url{http://www.graph500.org/}.

\bibitem[Pinar et~al.(2012)Pinar, Seshadhri, and Kolda]{Pinar2012}
A.~Pinar, C.~Seshadhri, and T.~G. Kolda, \emph{The Similarity between
  Stochastic Kronecker and Chung-Lu Graph Models}, 2012.

\bibitem[Leskovec(2008)]{Leskovec2008}
J.~Leskovec, ``Dynamics of large networks,'' Ph.D. dissertation, CMU, 2008.

\bibitem[Batagelj and Brandes(2005)]{Batagelj2005}
V.~Batagelj and U.~Brandes, ``{Efficient generation of large random
  networks},'' \emph{Phys. Rev. E}, 2005.

\bibitem[Miller and Hagberg(2011)]{Miller2011}
J.~Miller and A.~Hagberg, ``Efficient generation of networks with given
  expected degrees,'' in \emph{Proc. of Algorithms and Models for the
  Web-Graph}, 2011.

\bibitem[Manne and Sorevik(1995)]{Manne1995}
F.~Manne and T.~Sorevik, ``Optimal partitioning of sequences,'' \emph{J. of
  Algorithms}, 1995.

\bibitem[Olstad and Manne(1995)]{Olstad1995}
B.~Olstad and F.~Manne, ``Efficient partitioning of sequences,'' \emph{IEEE
  Trans. on Comput.}, 1995.

\bibitem[Pinar and Aykanat(2004)]{Pinar2004}
A.~Pinar and C.~Aykanat, ``Fast optimal load balancing algorithms for 1d
  partitioning,'' \emph{J. Par. Dist. Comp.}, 2004.

\bibitem[Sanders and Tr\"{a}ff(2006)]{Sanders2006}
P.~Sanders and J.~Tr\"{a}ff, ``Parallel prefix (scan) algorithms for mpi,'' in
  \emph{Proc. of the 13th Conf. on Rec. Adv. in PVM and MPI}, 2006.

\bibitem[Barrett et~al.(2009)Barrett, Beckman, Khan, Kumar, Marathe, Stretz,
  Dutta, and Lewis]{Barrett2009}
C.~Barrett, R.~Beckman, M.~Khan, V.~Kumar, M.~Marathe, P.~Stretz, T.~Dutta, and
  B.~Lewis, ``Generation and analysis of large synthetic social contact
  networks,'' in \emph{Proc. of the Winter Sim. Conf.}, 2009.

\end{thebibliography}

 \section{Appendix}
 \label{Section:Appendix}
 
 \begin{figure*}[ht!]
 \centering
 {\includegraphics[width=\textwidth]{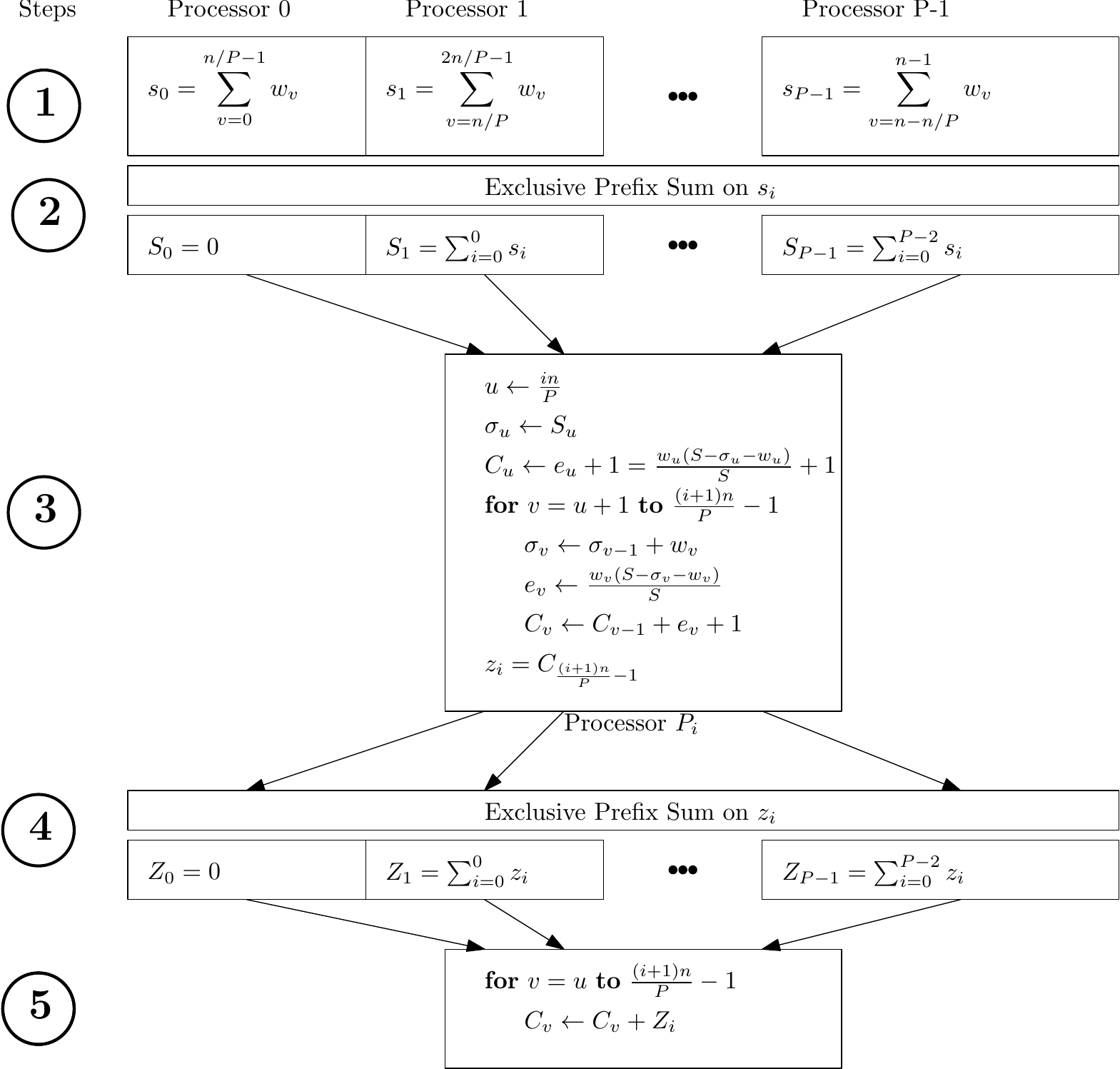}}
 \caption{Steps for determining cumulative cost in UCP}
 \label{Figure:PCL:uep-scheme}
 \end{figure*}
 
 \begin{figure*}[ht!]
 \centering
 {\includegraphics[width=\textwidth]{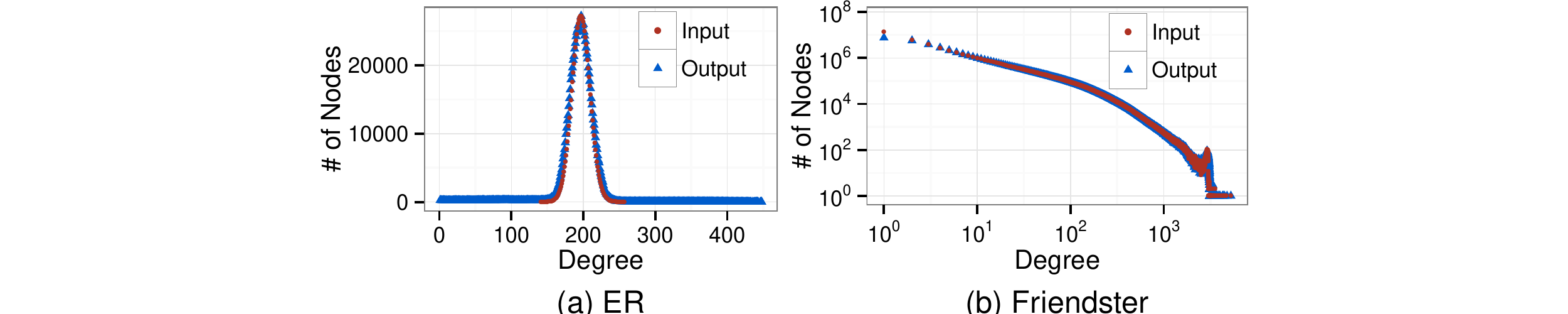}}
 \caption{Input and generated degree distributions for other networks}
 \label{Figure:PCL:OtherNetworks}
 \end{figure*}
 
\begin{lemma}
\label{lemma:eugeqev}
For any two nodes $u,v \in V$ such that $u<v$, $c_{u} \geq c_{v}$.
\end{lemma}
\begin{proof}
Proof omitted. The lemma follows immediately from Equation~\ref{Equation:PCL:cu} and the fact that, the weights are sorted in non-increasing order. 
\qed
\end{proof}

\begin{lemma}
\label{lemma:unp}
Let $c(V_i)$ be the computational cost for partition $V_i$. In the na\"{\i}ve partitioning scheme, we have $c(V_i) - c(V_{i+1}) \ge \frac{n^2}{SP^2} \overbar{W}_{i}\overbar{W}_{i+1}$, where $\overbar{W}_{i} = \frac{1}{|V_i|}\sum_{u\in V_i}{w_u}$, the average weight of the nodes in $V_i$.
\end{lemma}

\begin{proof}
In the na\"{\i}ve partitioning scheme, each of the partitions has $x = \frac{n}{P}$ nodes, except the last partition which can have smaller than $x$ nodes. For the ease of discussion, assume that for $u\ge n$, $w_u = 0$ and consequently $e_u = 0$. Now, $V_{i}=\left\{ ix, ix+1, \ldots, (i+1)x-1\right\}$.  Using Equation \ref{Equation:PCL:pcost}, we have

\begin{align*}
c(V_i)&- c(V_{i+1}) = \sum_{u \in V_{i}} (e_{u}+1) - \sum_{u \in V_{i+1}} (e_{u}+1) \\
\hspace{-1.5em}
\ge & \sum_{u=ix}^{(i+1)x - 1} (e_{u}+1) - \sum_{u=(i+1)x}^{(i+2)x - 1} (e_{u}+1) \\
= & \sum_{u=ix}^{(i+1)x - 1} (e_{u} - e_{u+x}) \\
= & \sum_{u=ix}^{(i+1)x - 1} \left( \frac{w_u}{S}\sum\limits_{v=u+1}^{n-1}w_{v} - \frac{w_{u+x}}{S}\sum\limits_{v=u+x+1}^{n-1}w_{v} \right)\\
\ge & \sum_{u=ix}^{(i+1)x - 1} \frac{w_u}{S}\sum\limits_{v=u+1}^{u+x}w_{v}
\ge  \sum_{u=ix}^{(i+1)x - 1} \frac{w_u}{S}x \overbar{W}_{i+1}\\
=&  \frac{x\overbar{W}_{i+1}}{S} \cdot x\overbar{W}_i 
=  \frac{n^2}{SP^2} \overbar{W}_{i}\overbar{W}_{i+1}
\end{align*}
\qed
\end{proof}

\begin{lemma}
\label{lemma:rrp}
In Round Robin Partitioning (RRP) scheme, for any $i<j$, we have $c(V_{i}) - c(V_{j}) \le w_{i}$.
\end{lemma}

\begin{proof}
The difference in  cost between two partitions $V_{i}$ and $V_{j}$ is given by:
\begin{align*}
c(V_{i}) - c(V_{j})  &= \sum_{u \in V_{i}}c_{u} - \sum_{u \in V_{j}}c_{u} =   \sum_{x=0}^{k}  \left(  c_{i+xP}-c_{j+xP} \right)\\
 &= c_{i}- \sum_{x=0}^{k-1} \left( c_{j+xP} - c_{i+(x+1)P} \right)  - c_{j+kP}\\
&\leq c_{i} - c_{j+kP}\text{\phantom{ABCDEF} {$\left[ c_{j+xp} \geq c_{i+(x+1)P} \right]$}}\\
&\leq e_{i} =\frac{w_i}{S}\sum_{v=i+1}^{n-1}w_v < \frac{w_i}{S}S =w_i
\end{align*}
\qed
\end{proof}

\subsection{Visual Representation of Computing Cost in UCP}
\label{Figure:Schematic Diagram}

\Figure~\ref{Figure:PCL:uep-scheme} shows the visual representation of \Call{Calc-Cost}{} procedure of Algorithm~\ref{Algorithm:PCL:UCP}.

\subsection{Other Networks}
\label{Section:OtherNetworks}

\Figure~\ref{Figure:PCL:OtherNetworks} shows input and generated degree distributions for ER and Friendster networks as shown in Table~\ref{Table:Networks}.


\end{document}